\documentclass[conference,letterpaper]{IEEEtran}

\usepackage[utf8]{inputenc}
\usepackage[T1]{fontenc}
\usepackage[cmex10]{amsmath}      
\usepackage{amssymb,amsfonts}
\usepackage{mathrsfs}             
\usepackage{cite}
\usepackage{graphicx}
\usepackage{url}
\usepackage{ifthen}
\usepackage{enumitem}             
\newlist{tightlist}{itemize}{3}
\setlist[tightlist]{label=$\bullet$, nosep}

\newcounter{theorem}
\newcounter{lemma}
\newcounter{proposition}
\newcounter{corollary}
\newcounter{definition}
\newcounter{example}
\newcounter{remark}

\newenvironment{theorem}[1][]{%
  \refstepcounter{theorem}\par\smallskip\noindent%
  \textbf{\textit{Theorem~\thetheorem}}%
  \ifthenelse{\equal{#1}{}}{}{ \textit{(#1)}}%
  \textit{.}\ \itshape\ignorespaces%
}{\par\smallskip}

\newenvironment{lemma}[1][]{%
  \refstepcounter{lemma}\par\smallskip\noindent%
  \textbf{\textit{Lemma~\thelemma}}%
  \ifthenelse{\equal{#1}{}}{}{ \textit{(#1)}}%
  \textit{.}\ \itshape\ignorespaces%
}{\par\smallskip}

\newenvironment{proposition}[1][]{%
  \refstepcounter{proposition}\par\smallskip\noindent%
  \textbf{\textit{Proposition~\theproposition}}%
  \ifthenelse{\equal{#1}{}}{}{ \textit{(#1)}}%
  \textit{.}\ \itshape\ignorespaces%
}{\par\smallskip}

\newenvironment{corollary}[1][]{%
  \refstepcounter{corollary}\par\smallskip\noindent%
  \textbf{\textit{Corollary~\thecorollary}}%
  \ifthenelse{\equal{#1}{}}{}{ \textit{(#1)}}%
  \textit{.}\ \itshape\ignorespaces%
}{\par\smallskip}

\newenvironment{definition}[1][]{%
  \refstepcounter{definition}\par\smallskip\noindent%
  \textbf{\textit{Definition~\thedefinition}}%
  \ifthenelse{\equal{#1}{}}{}{ \textit{(#1)}}%
  \textit{.}\ \normalfont\ignorespaces%
}{\par\smallskip}

\newenvironment{remark}[1][]{%
  \refstepcounter{remark}\par\smallskip\noindent%
  \textbf{\textit{Remark~\theremark}}%
  \ifthenelse{\equal{#1}{}}{}{ \textit{(#1)}}%
  \textit{.}\ \normalfont\ignorespaces%
}{\par\smallskip}

\newcommand{\F}{\mathbb{F}}
\newcommand{\dH}{d_H}
\newcommand{\wt}{\mathrm{wt}}
\newcommand{\Enc}{\mathrm{Enc}}
\newcommand{\pC}{\mathbf{p}_{\mathcal{C}}}
\newcommand{\CH}{\mathscr{C}_{H}}

\newcommand{\vv}[1]{\mathbf{#1}}

\begin{document}

\title{Function-Correction with Optimal Data Protection
       for the General Hamming Code Membership}
\author{%
  \IEEEauthorblockN{Adityawardhan Yadava}
  \IEEEauthorblockA{Department of Engineering Science\\
                    IIT Hyderabad
                    Hyderabad, India\\
                    Email: es23btech11001@iith.ac.in}
  \and
  \IEEEauthorblockN{Anjana A.\ Mahesh and Swaraj Sharma Durgi}
  \IEEEauthorblockA{Department of Electrical Engineering\\
                    IIT Hyderabad, 
                    Hyderabad, India\\
                    Email: anjana.am@ee.iith.ac.in, ee24btech11018@iith.ac.in}
}
\maketitle
\sloppy
\hbadness=10000
\vbadness=10000
\begin{abstract}
This paper investigates single-error-correcting function-correcting
codes~(SEFCCs) for the $[2^n-1,\,2^n-1-n,\,3]$-Hamming code
membership function~(HCMF) for general $n\geq 2$. Necessary and
sufficient conditions for valid parity assignments are established,
and the distance-$3$ codeword graph is shown to induce a connected
bipartite structure for all $n\geq 2$, which is exploited to develop
a systematic SEFCC construction achieving the largest possible minimum
distance of~$2$. A novel framework is then developed that reduces the
minimization of distance-$2$ codeword pairs to a max-cut problem on
the distance-$4$ graphs of the two partite sets. Eigenvectors
corresponding to the minimum eigenvalue of these graphs are shown to
directly yield optimal parity assignments. We reduce the problem of
finding these eigenvectors to an optimization problem involving
moments of the Walsh coefficients of a related function, which we
solve for even~$n$ by deriving a tight lower bound shown to be
attained by bent functions, establishing a precise connection between
optimal SEFCC design and bent Boolean functions.
\end{abstract}


\section{Introduction}
\label{sec:intro}
In many modern communication and distributed computing
systems, the receiver's objective is not the full reconstruction
of the transmitted message, but rather the reliable computation
of a specific function of that message. Classical error-correcting
codes (ECCs), designed to protect the entire message, often
provide unnecessary redundancy in such settings. This
observation motivated the introduction of function-correcting
codes (FCCs) by Lenz et al. ~\cite{lenz2023},
which guarantee that a function $f(\mathbf{u})$ can be correctly
recovered even when the codeword representing $\mathbf{u}$ is
corrupted by up to $t$ errors. FCCs are equivalent to
irregular-distance codes---codes that obey non-uniform,
function-dependent distance requirements between pairs of
codewords---and this equivalence enables tight bounds on
optimal redundancy to be derived. Since codewords assigned
to messages with the same function value carry no mutual
distance requirement, the redundancy of an FCC can be
substantially lower than that of a classical ECC protecting
the full message.

Since their introduction, FCCs have been studied along
several directions. Ge~et~al.\ analysed optimal redundancy
for Hamming weight and weight-distribution functions and
provided explicit constructions via a connection to Gray
codes~\cite{ge2025}. Premlal and Rajan established
new bounds on optimal FCC redundancy and extensively
studied the linear function case~\cite{premlal2025}.
The FCC framework has also been extended to other channel
models: Xia~et~al.\ generalised FCCs to the symbol-pair
metric~\cite{xia2024}, while Singh~et~al.\ extended
the framework to $b$-symbol read channels over finite
fields~\cite{singh2025}. Structural aspects---including
homogeneous distance requirements and locally bounded
function constraints---have been studied more recently
as well~ \cite{liu2025hom}, \cite{rajput2025lb}.

A key practical observation, highlighted in~\cite{rajput2025dp},
is that multiple FCCs may exist with the same redundancy
and function-correcting capability yet offer varying levels
of protection for the underlying data, implying that error
performance with respect to data bits is not solely determined
by the function-correction requirement. This motivated a
general framework for FCCs with data protection: Rajput et al. proposed a framework
that offers simultaneous protection for both the data and
the function value, focusing on scenarios where function
values require stronger protection than the data itself,
and demonstrated that data protection can sometimes be
incorporated into existing FCCs without any increase in
redundancy~\cite{rajput2025dp}.

Despite this progress, functions defined by algebraic
membership in a subspace of a vector space have received
limited attention. The Hamming code membership function
(HCMF)---which evaluates to $1$ if a given vector is a
codeword of a Hamming code and to $0$ otherwise---is a
canonical and practically motivated instance of this class.
It arises naturally in syndrome-based detection, coded
distributed computing, and storage systems where the
primary task is to verify membership rather than reconstruct
a full message. The prior work~\cite{durgi2026hcmf} addressed
this gap for the specific case of the $[7,4,3]$-Hamming code.
That work established necessary and sufficient conditions
for valid parity assignments in terms of distance constraints
between codewords and their nearest non-codewords, showed
that Hamming-distance-$3$ relations among the $16$ Hamming
codewords induce a bipartite graph, and exploited this
bipartite structure to develop a systematic SEFCC
construction. By deriving a tight upper bound on the sum
of pairwise distances, it was proved that the proposed
bipartite construction uniquely achieves the maximum
sum-distance, the largest possible minimum distance of $2$,
and the minimum number of distance-$2$ codeword pairs,
establishing that sum-distance maximisation is not merely
heuristic but an exact optimality criterion for the HCMF
SEFCC problem.

In this paper, we extend this line of work to the full family of
binary Hamming codes, studying single-error-correcting FCCs
for the $[2^n-1,\,2^n-1-n,\,3]$-Hamming code membership
function for general $n \geq 2$. The mechanism developed in~\cite{durgi2026hcmf}
for constructing SEFCCs with optimal data protection for the
$[7,4,3]$-HCMF relies on a quadratic minimization argument
that is specific to the small size and tight combinatorial struc-
ture of that case, and does not extend to the general HCMF.
We therefore develop a novel framework that characterises
SEFCCs with optimal data protection for the general HCMF
as a graph-theoretic optimisation problem. As a consistency
check, the solution to this optimisation problem specialised to
$n=3$ subsumes the code construction presented in~\cite{durgi2026hcmf}.
The contributions of this paper are as follows
\begin{itemize}

  \item We extend the necessary and sufficient conditions
  on parity assignments for a valid SEFCC derived in \cite{durgi2026hcmf} for the [7,4,3]-HCMF for the general
  Hamming code membership function, and prove that the
  distance-$3$ graph on the Hamming codewords retains a
  bipartite structure for all $n \geq 2$.

  \item A systematic SEFCC construction for the
  $[2^n-1,\,2^n-1-n,\,3]$-HCMF achieving optimal data
  protection is obtained by minimising the number of
  codeword pairs at the largest possible minimum distance of $2$. We pose
  this as a max-cut problem on the distance-$4$ graphs of
  the two partite sets independently, and show that
  eigenvectors corresponding to the minimum eigenvalue of
  the distance-$4$ graph on one of the partite sets directly yield
  optimal parity assignments for both sets.

  \item We formulate an optimisation problem involving
  the Walsh coefficients of a function that parameterises
  the eigenvalues of this graph and their corresponding eigenvectors. To achieve the minimum eigenvalue, a tight lower bound is derived and shown to be attained for even $n$ by bent functions, establishing a precise
  connection between optimal SEFCC design and the theory
  of bent Boolean functions.


\end{itemize}

The remainder of this paper is organised as follows.
Section~\ref{sec:prelim} provides the necessary
preliminaries on function-correcting codes and the
family of binary Hamming codes. Section~\ref{sec:setup}
defines the problem setup and the general Hamming code
membership function. Section~\ref{sec:main} presents
our main results: necessary and sufficient conditions
for valid parity assignment, the bipartite structure of
the distance-$3$ codeword graph, the max-cut formulation,
and the optimal construction via bent functions. Section~\ref{sec:conc} concludes the paper.


\textit{Notations:} For any positive integer $n$, the shorthand $[n]$ refers
to the set $\{1,2,\ldots,n\}$. The finite field with $2$ elements
is denoted by $\F_2$. Vectors of length $n$ over $\F_2$ form the vector
space $\F_2^n$. For a binary vector $\vv{x}$, its bit-wise complement is
denoted as $\bar{\vv{x}}$. A linear block code with length $n$, dimension
$k$, and minimum Hamming distance $d$ over $\F_2$ is denoted as
a $[n,k,d]$ code. The Hamming weight of a vector $\vv{x} \in \F_2^n$,
denoted $\wt(\vv{x})$, is the number of positions in $\vv{x}$ that are nonzero.
The Hamming distance between two vectors $\vv{x}, \vv{y} \in \F_2^n$,
denoted $\dH(\vv{x},\vv{y})$, is the number of positions where the bits
in $\vv{x}$ and $\vv{y}$ differ.

\section{Preliminaries}
\label{sec:prelim}

\subsection{Function-Correcting Codes}

\begin{definition}[Systematic Function-Correcting Code {\cite{lenz2023}}]
Let $f : \F_2^k \to \mathscr{I}$ be a function.
A systematic encoder $\Enc : \F_2^k \to \F_2^{k+r}$ of the form
$\Enc(\vv{u}) = (\vv{u},\, \vv{p}(\vv{u}))$ is an $(f,t)$-\emph{FCC} if for all
$\vv{u}_1, \vv{u}_2 \in \F_2^k$ with $f(\vv{u}_1) \neq f(\vv{u}_2)$,
\begin{equation}
  \dH\!\left(\Enc(\vv{u}_1),\Enc(\vv{u}_2)\right) \geq 2t+1.
  \label{eq:fcc_condition}
\end{equation}
\end{definition}

\begin{definition}[Optimal Redundancy {\cite{lenz2023}}]
Given $f$ and $t$, the \emph{optimal redundancy} $r_f(k,t)$ is the
minimum $r$ such that there exists an encoder
$\Enc : \F_2^k \to \F_2^{k+r}$ forming an $(f,t)$-FCC.
\end{definition}

\subsection{The $[2^n-1,\,2^n-1-n,\,3]$ Hamming Code}

The $[2^n-1,\,2^n-1-n,\,3]$-Hamming code is a binary linear block code
that encodes $2^n-1-n$ information bits into $2^n-1$ coded bits and has
minimum Hamming distance~$3$. Being a perfect code, meaning that every
vector in $\F_2^{2^n-1}$ lies within Hamming distance one of a unique
codeword, it partitions the space $\F_2^{2^n-1}$ into $2^{2^n-1-n}$
Hamming spheres, denoted $S_i$, $i \in [2^{2^n-1-n}]$, with one of the
$2^{2^n-1-n}$ codewords, denoted $\vv{c}_i$, at the centre of each sphere
and $2^n-1$ unique non-codewords at a Hamming distance of $1$ from the
codeword at the centre. We denote the non-codewords nearest to the codeword
$\vv{c}_i$ as $\vv{u}_{i,j}$, $j \in [2^n-1]$. Thus, the vector space
$\F_2^{2^n-1}$ is partitioned into
\begin{align*}
  \CH &= \{\vv{c}_i\}_{i=1}^{2^{2^n-1-n}}, \text{ and} \\
  \F_2^{2^n-1} \setminus \CH
      &= \{\vv{u}_{i,j} : i \in [2^{2^n-1-n}],\; j \in [2^n-1]\}.
\end{align*}
Owing to its algebraic structure, uniform distance properties, and
optimality for single-error correction, the
$[2^n-1,\,2^n-1-n,\,3]$-Hamming code serves as a canonical example in
coding theory and provides a natural testbed for studying FCCs based on
code membership functions~\cite{macwilliams1977}.

\section{Problem Setup}
\label{sec:setup}

For the $[2^n-1,\,2^n-1-n,\,3]$-Hamming code defined in
Section~\ref{sec:prelim}, let $\CH$ denote the set of $2^{2^n-1-n}$
codewords. We study function-correcting codes for the Hamming code
membership function (HCMF) $f_H : \F_2^{2^n-1} \to \F_2$ defined as
\begin{equation*}
  f_H(\vv{v}) =
  \begin{cases}
    1, & \vv{v} \in \CH, \\
    0, & \vv{v} \in \F_2^{2^n-1} \setminus \CH.
  \end{cases}
\end{equation*}
Thus, the domain $\F_2^{2^n-1}$ contains $2^{2^n-1-n}$ Hamming codewords
(HCWs) for which $f_H$ evaluates to $1$, and $2^{2^n-1} - 2^{2^n-1-n}$
vectors that are not Hamming codewords (NHCWs) for which $f_H$ evaluates
to~$0$.

For all boolean functions, it was shown in~\cite{lenz2023} that the optimal
redundancy for a $t$-error-correcting FCC is $r_f(k,t) = 2t$. Hence, an
optimal $t$-error-correcting systematic FCC $\mathcal{C}$ for the HCMF
$f_H$ is a mapping $\F_2^{2^n-1} \to \F_2^{2^n+1}$ with
$\Enc(\vv{u}_i) = [\vv{u}_i,\, \pC(\vv{u}_i)]$, where
$\pC(\vv{u}_i) \in \F_2^{2t}$ is the $2t$-length parity vector assigned
to the input vector $\vv{u}_i \in \F_2^{2^n-1}$. In this paper, we
consider single-error-correcting function-correcting codes (SEFCCs) for
the function $f_H$. For single-error correction ($t = 1$), the FCC
distance constraint in~\eqref{eq:fcc_condition} reduces to
\begin{equation}
  \dH\!\left(\Enc(\vv{c}_i),\Enc(\vv{u}_j)\right) \geq 3,
  \quad \forall\, \vv{c}_i \in \CH,\; \vv{u}_j \in \F_2^{2^n-1} \setminus \CH.
  \label{eq:sefcc_condition}
\end{equation}

Apart from the FCC distance constraint~\eqref{eq:sefcc_condition}, among
all valid parity assignments we further seek to: maximize 
\begin{equation}
  d_{\min}(\mathcal{C})
  \;=\; \min_{\substack{\vv{u},\vv{v}\,\in\,\F_2^{2^n-1} \\ \vv{u} \neq \vv{v}}}
        \dH\!\left(\Enc(\vv{u}),\Enc(\vv{v})\right),
  \label{eq:opt_dmin}
\end{equation}
and, among all valid parity assignments achieving the maximum
$d_{\min}(\mathcal{C})$, say $d_{\min}^*(\mathcal{C})$,  in~\eqref{eq:opt_dmin}: minimize the number of codewords at minimum distance i.e.,
\begin{equation}
  \bigl|\bigl\{\{\vv{u},\vv{v}\} : \dH\!\left(\Enc(\vv{u}),\Enc(\vv{v})\right)
  = d_{\min}^*(\mathcal{C})\bigr\}\bigr|.
  \label{eq:opt_pairs}
\end{equation}

\begin{definition}[Valid SEFCC]
\label{def:valid_sefcc}
An SEFCC $\mathcal{C} : \F_2^{2^n-1} \to \F_2^{2^n+1}$
for the HCMF $f_H$ is said to be \emph{valid} if it
satisfies~\eqref{eq:sefcc_condition}.
\end{definition}

\begin{definition}[Optimal SEFCC]
\label{def:optimal_sefcc}
A valid SEFCC $\mathcal{C} : \F_2^{2^n-1} \to \F_2^{2^n+1}$
for the HCMF $f_H$ is said to be \emph{optimal} and to have
\emph{optimal parity assignments} if it satisfies
constraints~\eqref{eq:opt_dmin} and~\eqref{eq:opt_pairs}.
\end{definition}

\section{Main Results}
\label{sec:main}

\subsection{Validity Conditions and Distance-3 Graph Structure}

Theorem \ref{thm:valid_dmin2} and Proposition \ref{prop:bipartite} below extends the corresponding results of~\cite{durgi2026hcmf} from the
$[7,4,3]$-HCMF to the general $[2^n-1,\,2^n-1-n,\,3]$-HCMF; the
proofs follow analogously.

\begin{theorem}[Valid SEFCC with Maximum $d_{\min}$]
\label{thm:valid_dmin2}
An SEFCC $\mathcal{C} : \F_2^{2^n-1} \to \F_2^{2^n+1}$ for
the HCMF $f_H$ is valid and achieves the maximum minimum distance
$d_{\min}^*(\mathcal{C}) = 2$ if and only if:
\begin{enumerate}[nosep, topsep=0pt, partopsep=0pt]
\item[a)] $\pC(\vv{u}_{i,j}) = \overline{\pC(\vv{c}_i)}$,
$\forall\, j \in [2^n-1]$, for each $i \in [2^{2^n-1-n}]$;
\item[b)] $\dH(\pC(\vv{c}_i), \pC(\vv{c}_k)) = 1$ for any pair
$(\vv{c}_i, \vv{c}_k)$ such that $\dH(\vv{c}_i, \vv{c}_k) = 3$,
$i, k \in [2^{2^n-1-n}]$.
\end{enumerate}
\end{theorem}

\begin{proposition}[Bipartite Structure of Distance-3 HCW Graph]
\label{prop:bipartite}
Let $G_3 = (\mathcal{V}, \mathcal{E})$ be the graph where vertex set $\mathcal{V} = \CH$ is the set of all codewords of the $[2^n-1,\,2^n-1-n,\,3]$ Hamming code and edge set $\mathcal{E} = \{(\vv{c}_i,\vv{c}_j) \mid \dH(\vv{c}_i,\vv{c}_j)=3\}$.
Then $G_3$ is bipartite with bipartition $\CH = \mathcal{C}_e \cup \mathcal{C}_o$,
where $\mathcal{C}_e$ and $\mathcal{C}_o$ denote the subsets of Hamming
codewords with even and odd Hamming weights, respectively.
\end{proposition}

\begin{IEEEproof}
We know that for a linear code, $\dH(\vv{c}_1,\vv{c}_2) =
\wt(\vv{c}_1 \oplus \vv{c}_2)$. If $\dH(\vv{c}_1,\vv{c}_2)=3$, then
$\wt(\vv{c}_1 \oplus \vv{c}_2)$ is odd, which implies $\wt(\vv{c}_1)$
and $\wt(\vv{c}_2)$ have opposite parity. Thus, partitioning the
codewords of a $[2^n-1,\,2^n-1-n,\,3]$-Hamming code by even/odd
weight yields a bipartition.
\end{IEEEproof}

We next show that $G_3$ is connected. The proof relies on the
following path lemma.

\begin{lemma}[Path Lemma]
\label{lem:path}
Let $\vv{v} \in \CH$ and let $\vv{v} = \vv{\alpha}_1 \oplus
\vv{\alpha}_2 \oplus \cdots \oplus \vv{\alpha}_k$ where
$\vv{\alpha}_i \in \CH$ and $\wt(\vv{\alpha}_i) = 3$ for all
$i \in [k]$. Define $\vv{\beta}_0 = \vv{0}$ and
$\vv{\beta}_i = \vv{\alpha}_1 \oplus \cdots \oplus \vv{\alpha}_i$
for $i \in [k]$. Then $\vv{\beta}_0 \to \vv{\beta}_1 \to \cdots
\to \vv{\beta}_k$ is a valid path in $G_3$.
\end{lemma}

\begin{IEEEproof}
Since the $[2^n-1,\,2^n-1-n,\,3]$-Hamming code is linear and
$\vv{\alpha}_i \in \CH$ for all $i$, each $\vv{\beta}_i$ is a sum
of codewords and hence a codeword, so every vertex in the proposed
path is a valid vertex of $G_3$. Consecutive vertices are adjacent
since
\begin{equation*}
  \dH(\vv{\beta}_i, \vv{\beta}_{i+1})
  = \wt(\vv{\beta}_i \oplus \vv{\beta}_{i+1})
  = \wt(\vv{\alpha}_{i+1}) = 3, \quad \forall\, i < k,
\end{equation*}
confirming that each consecutive pair $(\vv{\beta}_i,
\vv{\beta}_{i+1})$ is an edge in $G_3$.
\end{IEEEproof}

\begin{proposition}[Connectivity of $G_3$]
\label{prop:connected}
The distance-$3$ HCW graph $G_3$ of the
$[2^n-1,\,2^n-1-n,\,3]$-Hamming code is connected.
\end{proposition}

\begin{IEEEproof}
Let $W = \{\vv{c} \in \CH \mid \wt(\vv{c}) = 3\}$ be the set of
minimum-weight codewords. It is well known that the minimum-weight
codewords of a Hamming code generate the entire
code~\cite{macwilliams1977}, i.e., $\mathrm{span}(W) = \CH$.
Hence, for any $\vv{v} \in \CH$, there exist
$\vv{\alpha}_1, \ldots, \vv{\alpha}_k \in W$ such that
$\vv{v} = \vv{\alpha}_1 \oplus \cdots \oplus \vv{\alpha}_k$.
By Lemma~\ref{lem:path}, there exists a path
$\vv{0} \to \vv{\beta}_1 \to \cdots \to \vv{\beta}_k = \vv{v}$
in $G_3$, so every $\vv{v} \in \CH$ is reachable from $\vv{0}$.
Since reachability is symmetric and transitive, any
$\vv{v}, \vv{w} \in \CH$ are connected via $\vv{0}$, and $G_3$
is therefore connected.
\end{IEEEproof}

By Proposition~\ref{prop:bipartite}, $\CH$ admits a bipartition
$\CH = \mathcal{C}_o \cup \mathcal{C}_e$, where $\mathcal{C}_o$ and
$\mathcal{C}_e$ denote the subsets of Hamming codewords with odd and
even Hamming weights, respectively. Define two complementary
parity-pair sets
\[
  P_e \triangleq \{00,\,11\} \quad \text{and} \quad
  P_o \triangleq \{01,\,10\}.
\]
Note that $\dH(\vv{p},\vv{p}') = 1$ for any $\vv{p} \in P_e$ and
$\vv{p}' \in P_o$, so any cross-pair assignment satisfies the
distance-$1$ parity requirement of Theorem~\ref{thm:valid_dmin2}(b).

\begin{theorem}[Bipartite Parity Structure]
\label{thm:bipartite_parity}
A parity assignment for the $2^{2^n-1-n}$ Hamming codewords
constitutes a valid optimal SEFCC
$\mathcal{C} : \F_2^{2^n-1} \to \F_2^{2^n+1}$ achieving
$d_{\min}^*(\mathcal{C}) = 2$ if and only if the parities assigned
to all codewords in $\mathcal{C}_o$ are drawn exclusively from $P_o$
and the parities assigned to all codewords in $\mathcal{C}_e$ are
drawn exclusively from $P_e$ (or equivalently, with the roles of
$P_o$ and $P_e$ interchanged).
\end{theorem}

\begin{IEEEproof}
$(\Rightarrow)$: Let $\mathcal{C}$ be a valid optimal SEFCC achieving
$d_{\min}^*(\mathcal{C}) = 2$. Without loss of generality, suppose
$\vv{0} \in \mathcal{C}_e$ is assigned a parity from $P_e$. By
Theorem~\ref{thm:valid_dmin2}(b), every distance-$3$ neighbour of
$\vv{0}$ in $\mathcal{C}_o$ must be assigned a parity at Hamming
distance $1$ from that of $\vv{0}$, forcing their parities into $P_o$.
Their distance-$3$ neighbours in $\mathcal{C}_e$ are in turn forced
into $P_e$, and so on. Since $G_3$ is connected
(Proposition~\ref{prop:connected}), this propagation---equivalent to
a breadth-first traversal of $G_3$ from $\vv{0}$---reaches every
codeword: vertices at even distance from $\vv{0}$ lie in
$\mathcal{C}_e$ and receive parities from $P_e$, while vertices at
odd distance lie in $\mathcal{C}_o$ and receive parities from $P_o$.

$(\Leftarrow)$: Conversely, if all codewords in $\mathcal{C}_o$
are assigned parities from $P_o$ and all codewords in $\mathcal{C}_e$
from $P_e$, then for any pair $(\vv{c}_i, \vv{c}_k)$ with
$\dH(\vv{c}_i,\vv{c}_k) = 3$, the two codewords lie in different
partite sets and are thus assigned parities from different sets
$P_e$ and $P_o$. Since $\dH(\vv{p},\vv{p}') = 1$ for any
$\vv{p} \in P_e,\, \vv{p}' \in P_o$, the condition of
Theorem~\ref{thm:valid_dmin2}(b) is satisfied, and the assignment
constitutes a valid optimal SEFCC with $d_{\min}^*(\mathcal{C}) = 2$.
\end{IEEEproof}

\subsection{Minimizing Distance-$2$ Pairs via Spectral Graph Theory}
\label{subsec:spectral}

Theorems~\ref{thm:valid_dmin2}a) and~\ref{thm:bipartite_parity} together
completely characterize valid optimal SEFCCs
$\mathcal{C}:\F_2^{2^n-1}\to\F_2^{2^n+1}$ achieving
$d_{\min}^*(\mathcal{C})=2$. Assuming a parity assignment following
Theorem~\ref{thm:bipartite_parity}, we now turn to
constraint~\eqref{eq:opt_pairs}: minimizing the number of unordered pairs
$\{\vv{u},\vv{v}\}\subset\F_2^{2^n-1}$ satisfying
$\dH(\Enc(\vv{u}),\Enc(\vv{v}))=2$. Such distance-$2$ pairs originate from exactly three sources:
\begin{enumerate}[nosep, topsep=0pt, partopsep=0pt]
\item[(i)] \textbf{Intra-sphere pairs:} Two NHCWs $\vv{u}_{i,a},\vv{u}_{i,b}$
within the same sphere $S_i$. Here $\dH(\vv{u}_{i,a},\vv{u}_{i,b})=2$ and
the parity distance is $0$ as both share $\overline{\pC(\vv{c}_i)}$.
There are $2^{2^n-1-n}\times\binom{2^n-1}{2}$ such pairs, invariant under
any valid parity assignment.

\item[(ii)] \textbf{Inter-sphere boundary pairs:} NHCWs
$\vv{u}_{i,a}\in S_i$ and $\vv{u}_{k,b}\in S_k$ with
$\dH(\vv{c}_i,\vv{c}_k)=3$, for which $\dH(\vv{u}_{i,a},\vv{u}_{k,b})=1$
and $\dH(\pC(\vv{c}_i),\pC(\vv{c}_k))=1$ (since parities are drawn from
disjoint sets $P_e$ and $P_o$). This count is invariant under any parity
assignment satisfying Theorem~\ref{thm:bipartite_parity} and Theorem~\ref{thm:valid_dmin2}(a)

\item[(iii)] \textbf{Identical-parity pairs:} NHCWs $\vv{u}_{i,a}\in S_i$
and $\vv{u}_{k,b}\in S_k$ with $\dH(\vv{c}_i,\vv{c}_k)=4$ and
$\pC(\vv{c}_i)=\pC(\vv{c}_k)$. Here $\dH(\vv{u}_{i,a},\vv{u}_{k,b})=2$
and the parity distance is $0$.
\end{enumerate}

Since categories~(i) and~(ii) are invariant, minimizing the total count of
distance-$2$ pairs reduces to minimizing the count in category~(iii). One
can verify that each pair $\{\vv{c}_i,\vv{c}_k\}$ with
$\dH(\vv{c}_i,\vv{c}_k)=4$ and $\pC(\vv{c}_i)=\pC(\vv{c}_k)$ contributes
exactly $12$ identical-parity NHCW pairs at distance~$2$.

We next show that any HCW pair at Hamming distance~$4$ must lie within the
same partite set. Suppose $\vv{c}_i\in\mathcal{C}_e$ and
$\vv{c}_k\in\mathcal{C}_o$ with $\dH(\vv{c}_i,\vv{c}_k)=4$. Then
$\wt(\vv{c}_i\oplus\vv{c}_k)=4$ is even, which forces
$\wt(\vv{c}_i)\equiv\wt(\vv{c}_k)\pmod{2}$, contradicting the assumption
that $\vv{c}_i$ and $\vv{c}_k$ lie in opposite partite sets. Hence all
distance-$4$ HCW pairs lie entirely within $\mathcal{C}_e$ or within
$\mathcal{C}_o$. This means that after assigning $\mathcal{C}_e$ a complementary parity
pair set $P\in\{P_e,P_o\}$ and assigning $\mathcal{C}_o$ the remaining
set $\{P_e,P_o\}\setminus P$, we can independently minimize the number
of same-parity pairs at distance~$4$ within $\mathcal{C}_e$ and
$\mathcal{C}_o$, since no distance-$4$ edge crosses the bipartition. We thus focus exclusively on finding the optimal parity assignemnts for $\mathcal{C}_e$ for now.  We identify the structure of the
distance-$4$ graph on $\mathcal{C}_e$ and formulate the minimization as a
max-cut problem, whose solution is characterized by eigenvectors of the
minimum eigenvalue of the graph's adjacency matrix.

\begin{theorem}[Max-Cut Formulation]
\label{thm:maxcut}
Let $\vv{v}_1,\ldots,\vv{v}_{|\mathcal{C}_e|}$ be the HCWs in
$\mathcal{C}_e$, where $|\mathcal{C}_e|=2^{2^n-n-2}$. Let
$G_4=(\mathcal{C}_e,\mathcal{E}_4)$ be the graph with
$\{\vv{v}_i,\vv{v}_j\}\in\mathcal{E}_4$ iff $\dH(\vv{v}_i,\vv{v}_j)=4$,
and let $A$ be the adjacency matrix of $G_4$. Given that parity assignments to HCWs in $\mathcal{C}_e$ are drawn from $\{\vv{p},\bar{\vv{p}}\}\in\{P_e,P_o\}$, define the cut vector $\vv{z}\in\{+1,-1\}^{|\mathcal{C}_e|}$ by $z_i=+1$ if $\pC(\vv{v}_i)=\vv{p}$ and $z_i=-1$ if $\pC(\vv{v}_i)=\bar{\vv{p}}$. Then the number of same-parity pairs in $\mathcal{E}_4$ is minimized if and only if $\vv{z}$ minimizes $\vv{z}^\top A\vv{z}$.
\end{theorem}

\begin{IEEEproof}
An edge $\{\vv{v}_i,\vv{v}_j\}\in\mathcal{E}_4$ is a \emph{cut edge} if
$z_iz_j=-1$ (distinct parities) and a \emph{same-parity edge} if $z_iz_j=+1$.
Denoting the number of cut and same-parity edges by $|\mathcal{E}_c|$ and
$|\mathcal{E}_s|$, respectively, we have
\begin{equation*}
  \vv{z}^\top A\vv{z}
  = 2\!\!\!\sum_{\{\vv{v}_i,\vv{v}_j\}\in\mathcal{E}_4}\!\!\! z_iz_j
  = 2|\mathcal{E}_s| - 2|\mathcal{E}_c|.
\end{equation*}
Since $|\mathcal{E}_s|+|\mathcal{E}_c|=|\mathcal{E}_4|$, we obtain
$|\mathcal{E}_c| = |\mathcal{E}_4|/2 - \vv{z}^\top A\vv{z}/4$,
so $|\mathcal{E}_c|$ is maximized if and only if $\vv{z}^\top A\vv{z}$
is minimized.
\end{IEEEproof}

We now establish that $G_4$ is a Cayley graph on $(\mathcal{C}_e,\oplus)$,
a structure that enables explicit computation of the eigenvectors of~$A$.

\begin{lemma}[Cayley Graph Structure of $G_4$]
\label{lem:cayley}
Let $\mathcal{S}_4\triangleq\{\vv{v}\in\mathcal{C}_e\mid\wt(\vv{v})=4\}$.
Then $G_4=\mathrm{Cay}(\mathcal{C}_e,\oplus,\mathcal{S}_4)$, i.e., $G_4$
is the Cayley graph of the group $(\mathcal{C}_e,\oplus)$ with generating
set $\mathcal{S}_4$. The characters of $(\mathcal{C}_e,\oplus)$ are the
homomorphisms $\chi_{\vv{u}}:\mathcal{C}_e\to\{+1,-1\}$ defined by
$\chi_{\vv{u}}(\vv{v})=(-1)^{\vv{u}\cdot\vv{v}}$,
for $\vv{u}\in\F_2^{2^n-1}$.
\end{lemma}

\begin{IEEEproof}
Since the $[2^n-1,2^n-1-n,3]$-Hamming code is linear,
$(\mathcal{C}_e,\oplus)$ is an abelian group. We verify the Cayley graph
characterization. If $\vv{x}\sim\vv{y}$ in $G_4$, then
$\wt(\vv{x}\oplus\vv{y})=4$; since $\mathcal{C}_e$ is closed under $\oplus$,
$\vv{x}\oplus\vv{y}\in\mathcal{S}_4$. Conversely, if
$\vv{x}\in\mathcal{C}_e$ and $\vv{s}\in\mathcal{S}_4$, then
$\vv{y}=\vv{x}\oplus\vv{s}\in\mathcal{C}_e$ and
$\dH(\vv{x},\vv{y})=\wt(\vv{s})=4$, so $\vv{x}\sim\vv{y}$. 

For the character classification: since every $\vv{v}\in\mathcal{C}_e$
satisfies $\vv{v}\oplus\vv{v}=\vv{0}$, any character
$\chi:\mathcal{C}_e\to\mathbb{C}^\times$ satisfies
$\chi(\vv{v})^2=\chi(\vv{0})=1$, so $\chi(\vv{v})\in\{+1,-1\}$ for all
$\vv{v}$. Writing $\chi(\vv{v})=(-1)^{f(\vv{v})}$ for some
$f:\mathcal{C}_e\to\F_2$, the homomorphism condition
$\chi(\vv{x}\oplus\vv{y})=\chi(\vv{x})\chi(\vv{y})$ becomes
$f(\vv{x}\oplus\vv{y})\equiv f(\vv{x})+f(\vv{y})\pmod{2}$, so $f$ is a
linear functional on $\mathcal{C}_e$. Any linear functional on the subspace
$\mathcal{C}_e\subseteq\F_2^{2^n-1}$ extends to $\F_2^{2^n-1}$, where it
takes the form $f(\vv{v})=\vv{u}\cdot\vv{v}$ for some
$\vv{u}\in\F_2^{2^n-1}$~\cite{macwilliams1977}. Hence every character has
the form $\chi_{\vv{u}}(\vv{v})=(-1)^{\vv{u}\cdot\vv{v}}$.
\end{IEEEproof}

\begin{corollary}[Eigenvectors of $A$]
\label{cor:eigvecs}
For each $\vv{u}\in\F_2^{2^n-1}$, the vector
\begin{equation}
  \vv{e}_{\vv{u}} \triangleq
  \bigl[(-1)^{\vv{u}\cdot\vv{v}_1},\;(-1)^{\vv{u}\cdot\vv{v}_2},\;\ldots,\;
  (-1)^{\vv{u}\cdot\vv{v}_{|\mathcal{C}_e|}}\bigr]^\top
  \label{eq:eigvec}
\end{equation}
is an eigenvector of $A$ with eigenvalue
\begin{equation}
  \lambda_{\vv{u}} \triangleq \sum_{\vv{s}\in\mathcal{S}_4}(-1)^{\vv{u}\cdot\vv{s}}.
  \label{eq:eigenvalue}
\end{equation}
The collection $\mathcal{Q}\triangleq\{\vv{e}_{\vv{u}}\mid\vv{u}\in\F_2^{2^n-1}\}$
comprises $|\mathcal{C}_e|=2^{2^n-n-2}$ distinct, mutually orthogonal
vectors forming an eigenbasis for $\mathbb{R}^{|\mathcal{C}_e|}$
with respect to $A$. Denote the minimum eigenvalue by
$\lambda_{\min}\triangleq\min_{\vv{u}\in\F_2^{2^n-1}}\lambda_{\vv{u}}$.
\end{corollary}

\begin{IEEEproof}
This follows from the standard spectral theory of Cayley graphs on abelian
groups: the characters of the group are eigenvectors of the adjacency
matrix, with eigenvalue equal to the sum of the character over the
generating set~\cite{macwilliams1977}. The $|\mathcal{C}_e|$ distinct
characters of $(\mathcal{C}_e,\oplus)$ yield $|\mathcal{C}_e|$ mutually
orthogonal eigenvectors spanning $\mathbb{R}^{|\mathcal{C}_e|}$.
\end{IEEEproof}

\begin{theorem}[Optimal Parity Assignment via Eigenvectors]
\label{thm:eigvec_opt}
Any eigenvector $\vv{e}_{\vv{u}}\in\mathcal{Q}$ or its negative
$-\vv{e}_{\vv{u}}$ corresponding to $\lambda_{\min}$ is a valid cut
vector $\vv{z}\in\{+1,-1\}^{|\mathcal{C}_e|}$ and minimizes
$\vv{z}^\top A\vv{z}$. Such eigenvectors yield the optimal parity
assignment for $\mathcal{C}_e$, i.e., a parity assignment satisfying
constraint~\eqref{eq:opt_pairs}.
\end{theorem}

\begin{IEEEproof}
Since $A$ is real symmetric, the Rayleigh quotient gives
\begin{equation*}
  \vv{z}^\top A\vv{z} \geq \lambda_{\min}\|\vv{z}\|^2
  = \lambda_{\min}\,|\mathcal{C}_e|
\end{equation*}
for any $\vv{z}$, with equality if and only if $\vv{z}$ is an eigenvector
of $\lambda_{\min}$. To achieve this bound over
$\vv{z}\in\{+1,-1\}^{|\mathcal{C}_e|}$, the eigenvector must have entries
in $\{+1,-1\}$, i.e., it must be a valid cut vector. By
Corollary~\ref{cor:eigvecs}, every $\vv{e}_{\vv{u}}\in\mathcal{Q}$ has
entries $\pm 1$ by construction. Hence those $\vv{e}_{\vv{u}}\in\mathcal{Q}$ (or their negatives $-\vv{e}_{\vv{u}}$)
corresponding to $\lambda_{\min}$ are valid cut vectors achieving the lower
bound, and yield the optimal parity assignment for $\mathcal{C}_e$.
\end{IEEEproof}

\begin{remark}
Although any eigenvector of $\lambda_{\min}$ minimizes $\vv{z}^\top A\vv{z}$
over $\mathbb{R}^{|\mathcal{C}_e|}$, an arbitrary linear combination of eigenvectors in $\mathcal{Q}$ corresponding to $\lambda_{\min}$ need not have $\pm1$ entries and hence may
not define a valid parity assignment. The optimality result relies
specifically on the character vectors $\vv{e}_{\vv{u}}\in\mathcal{Q}$ (or their negatives $-\vv{e}_{\vv{u}}$),
which are guaranteed to be $\pm1$-valued.
\end{remark}
\subsection{Eigenvalue Analysis and Optimal Parity Assignments via Bent Functions}
\label{subsec:eigenvalues}
In this section, we exploit the algebraic structure of
$\mathcal{C}_e$ and its dual to reduce the computation of
$\lambda_{\min}$ and the corresponding optimal eigenvectors to an
optimization problem over the Walsh coefficients of a related function,
which we solve for even~$n$.
For notational convenience, let $q \triangleq 2^n$ throughout this
section; we substitute $q = 2^n$ back in the final result. We begin
with a lemma expressing the indicator function of $\mathcal{C}_e$ via
its dual code $\mathcal{C}_e^\perp$.

\begin{lemma}[Indicator Function of $\mathcal{C}_e$]
\label{lem:indicator}
For any $\vv{x} \in \F_2^{q-1}$, the indicator function for $\mathcal{C}_e$ can be expressed as 
\begin{equation}
  \mathbf{1}_{\mathcal{C}_e}(\vv{x})
  = \frac{1}{|\mathcal{C}_e^\perp|}
    \sum_{\vv{v}\,\in\,\mathcal{C}_e^\perp} (-1)^{\vv{v}\cdot\vv{x}}.
  \label{eq:indicator}
\end{equation}
\end{lemma}

\begin{IEEEproof}
If $\vv{x}\in\mathcal{C}_e$, then $\vv{v}\cdot\vv{x}=0$ for all
$\vv{v}\in\mathcal{C}_e^\perp$ by definition of the dual code, so
every term in the sum equals $1$ and the right-hand side equals $1$.
If $\vv{x}\notin\mathcal{C}_e$, the map $\vv{v}\mapsto\vv{v}\cdot\vv{x}$
is a non-trivial linear functional on $\mathcal{C}_e^\perp$ (since
$\vv{v}\cdot\vv{x}=0$ for all $\vv{v}\in\mathcal{C}_e^\perp$ would
imply $\vv{x}\in(\mathcal{C}_e^\perp)^\perp=\mathcal{C}_e$,
a contradiction). Any non-trivial character of a finite abelian group
sums to zero over the group, so
$\sum_{\vv{v}\in\mathcal{C}_e^\perp}(-1)^{\vv{v}\cdot\vv{x}}=0$ and
the right-hand side equals $0$.
\end{IEEEproof}
\begin{definition}[Truth Table]
\label{def:Truth}
The truth table of a function $f:\F_2^n\to\{0,1\}$ is the vector
$T=(f(\vv{x}_i))_{i=1}^{2^n}\in\F_2^{2^n}$, where $\vv{x}_i\in\F_2^n$
are in lexicographic order. Its punctured truth table in
$\F_2^{2^n-1}$ is obtained by removing the first bit, i.e., $f(\vv{0})$,
from $T$.
\end{definition}
\begin{proposition}[Eigenvalue Formula]
\label{prop:eigenvalue_formula}
Let $\vv{u}\in\F_2^{q-1}$ be the punctured truth table of a function
$f_{\vv{u}}:\F_2^n\to\{0,1\}$ satisfying $f_{\vv{u}}(\vv{0})=0$, and
let the Walsh coefficient of $f_{\vv{u}}$ at $\vv{a}\in\F_2^n$ be
defined as
\begin{equation}
  W_{f_{\vv{u}}}(\vv{a})
  \triangleq \sum_{\vv{x}\,\in\,\F_2^n}(-1)^{f_{\vv{u}}(\vv{x})+\vv{a}\cdot\vv{x}}.
  \label{eq:walsh}
\end{equation}
Then the eigenvalue $\lambda_{\vv{u}}$ of the adjacency matrix $A$ of
$G_4$ corresponding to the eigenvector $\vv{e}_{\vv{u}}$ satisfies
\begin{equation}
  24\,\lambda_{\vv{u}}
  = \frac{1}{2^n}\sum_{\vv{a}\,\in\,\F_2^n}
    \Bigl(W_{f_{\vv{u}}}(\vv{a})^4 - 4\,W_{f_{\vv{u}}}(\vv{a})^3\Bigr)
    - 3\cdot 2^{2n} + 14\cdot 2^n - 8.
  \label{eq:eigenvalue_formula}
\end{equation}
\end{proposition}

\begin{IEEEproof}
From Corollary~\ref{cor:eigvecs},
\begin{equation*}
  \lambda_{\vv{u}}
  = \sum_{\vv{s}\,\in\,\mathcal{S}_4}(-1)^{\vv{u}\cdot\vv{s}}
  = \sum_{\substack{\vv{x}\,\in\,\F_2^{q-1} \\ \wt(\vv{x})=4}}
    \mathbf{1}_{\mathcal{C}_e}(\vv{x})\,(-1)^{\vv{u}\cdot\vv{x}}.
\end{equation*}
Substituting Lemma~\ref{lem:indicator} and exchanging the order of
summation,
\begin{equation}
\begin{split}
  \lambda_{\vv{u}}
  &= \frac{1}{|\mathcal{C}_e^\perp|}
    \sum_{\vv{v}\,\in\,\mathcal{C}_e^\perp}
    \sum_{\substack{\vv{x}\,\in\,\F_2^{q-1} \\ \wt(\vv{x})=4}}
    (-1)^{(\vv{u}+\vv{v})\cdot\vv{x}} \\
  &= \frac{1}{|\mathcal{C}_e^\perp|}
    \sum_{\vv{v}\,\in\,\mathcal{C}_e^\perp}
    K_4\!\bigl(\wt(\vv{u}+\vv{v});\,N\bigr),
\end{split}
  \label{eq:lam_kraw}
\end{equation}
where $N=q-1$ and $K_4(w;\,N)$ denotes the Krawtchouk polynomial of
degree $4$, defined by
$K_k(w;\,N)=\sum_{j=0}^{k}(-1)^j\binom{w}{j}\binom{N-w}{k-j}$, and
the last equality follows from the standard identity
$\sum_{\wt(\vv{x})=k}(-1)^{\vv{w}\cdot\vv{x}}=K_k(\wt(\vv{w});\,N)$.

It is shown in Appendix~A that
$\mathcal{C}_e^\perp = \mathcal{RM}^*(1,n)$, the punctured first-order
Reed--Muller code of length $q-1$ and dimension $n+1$, so
$|\mathcal{C}_e^\perp|=2^{n+1}=2q$. The $2q$ codewords of
$\mathcal{RM}^*(1,n)$ are precisely the punctured truth tables of the
affine functions $\vv{a}\cdot\vv{x}+c$, parameterised by
$(\vv{a},c)\in\F_2^n\times\{0,1\}$; we denote the corresponding
codeword $\vv{v}_{\vv{a},c}$. Equation~\eqref{eq:lam_kraw} becomes
\begin{equation}
  \lambda_{\vv{u}}
  = \frac{1}{2q}
    \sum_{\vv{a}\,\in\,\F_2^n}\sum_{c\,\in\,\{0,1\}}
    K_4\!\bigl(\wt(\vv{u}+\vv{v}_{\vv{a},c});\,N\bigr).
  \label{eq:lam_affine}
\end{equation}

Since $\vv{u}$ is the punctured truth table of $f_{\vv{u}}$ and
$\vv{v}_{\vv{a},c}$ is the punctured truth table of
$\vv{a}\cdot\vv{x}+c$, the Hamming weight $\wt(\vv{u}+\vv{v}_{\vv{a},c})$
counts the number of positions (over $\F_2^n\setminus\{\vv{0}\}$) where
$f_{\vv{u}}$ and $\vv{a}\cdot\vv{x}+c$ disagree. Let
$D$ denote the number of disagreements over the full domain $\F_2^n$
and $A=q-D$ the number of agreements. Then
\begin{equation*}
  A - D
  = \sum_{\vv{x}\,\in\,\F_2^n}(-1)^{f_{\vv{u}}(\vv{x})+\vv{a}\cdot\vv{x}+c}
  = (-1)^c\,W_{f_{\vv{u}}}(\vv{a}),
\end{equation*}
which gives $D = 2^{n-1} - \tfrac{(-1)^c}{2}W_{f_{\vv{u}}}(\vv{a})$.
Since $f_{\vv{u}}(\vv{0})=0$, the position $\vv{x}=\vv{0}$ contributes
an \emph{agreement} when $c=0$ and a \emph{disagreement} when $c=1$.
Removing this position to pass to the punctured count $D'$:
\begin{align*}
  \wt(\vv{u}+\vv{v}_{\vv{a},0})
    &= 2^{n-1} - \tfrac{1}{2}W_{f_{\vv{u}}}(\vv{a}), \\
  \wt(\vv{u}+\vv{v}_{\vv{a},1})
    &= 2^{n-1} + \tfrac{1}{2}W_{f_{\vv{u}}}(\vv{a}) - 1.
\end{align*}
Note that the second expression equals $N - \wt(\vv{u}+\vv{v}_{\vv{a},0})$.
Applying the Krawtchouk symmetry property $K_4(w;\,N)=K_4(N-w;\,N)$
(since $(-1)^4=1$; see~\cite{macwilliams1977}), the two terms in the
inner sum of~\eqref{eq:lam_affine} are equal, giving
\begin{equation}
  \lambda_{\vv{u}}
  = \frac{1}{q}\sum_{\vv{a}\,\in\,\F_2^n}
    K_4\!\Bigl(\tfrac{q}{2}-\tfrac{W_{f_{\vv{u}}}(\vv{a})}{2};\,q-1\Bigr).
  \label{eq:lam_walsh}
\end{equation}
(Also note that the expressions for $\wt(\vv{u}+\vv{v}_{\vv{a},0})$
and $\wt(\vv{u}+\vv{v}_{\vv{a},1})$ imply that $W_{f_{\vv{u}}}(\vv{a})$
must be even for all $\vv{a}\in\F_2^n$, a fact that will be used in
the proof of Proposition~\ref{prop:lowerbound}.)
Setting $w=\tfrac{q}{2}-\tfrac{x}{2}$ and $N=q-1$ so that $X\triangleq N-2w=x-1$,
it is shown in Appendix~B that
\begin{equation*}
\begin{split}
  24\,K_4\!\Bigl(\tfrac{q}{2}-\tfrac{x}{2};\,q-1\Bigr)
  &= x^4 - 4x^3 + (20-6q)x^2 \\
  &\quad + (12q-32)x + (3q^2-18q+24).
\end{split}
\end{equation*}
Substituting $x=W_{f_{\vv{u}}}(\vv{a})$ and summing
over $\vv{a}\in\F_2^n$, we apply the identities
\begin{equation*}
  \sum_{\vv{a}\,\in\,\F_2^n} W_{f_{\vv{u}}}(\vv{a}) = q,
  \qquad
  \sum_{\vv{a}\,\in\,\F_2^n} W_{f_{\vv{u}}}(\vv{a})^2 = q^2,
\end{equation*}
where the first follows from $f_{\vv{u}}(\vv{0})=0$ and the second is
Parseval's identity for Walsh transforms~\cite{macwilliams1977}.
Collecting constant terms:
\begin{align*}
  24\lambda_{\vv{u}}
  &= \frac{1}{q}\sum_{\vv{a}}\bigl(W_{f_{\vv{u}}}(\vv{a})^4
     - 4\,W_{f_{\vv{u}}}(\vv{a})^3\bigr) \\
  &\quad + (20-6q)\,q + (12q-32) + (3q^2-18q+24).
\end{align*}
Simplifying the constant part:
$(20-6q)q + 12q - 32 + 3q^2 - 18q + 24
= -3q^2 + 14q - 8$,
and substituting $q = 2^n$ yields~\eqref{eq:eigenvalue_formula}.
\end{IEEEproof}

Thus the problem of finding the minimum eigenvalue of $A$ reduces
to minimizing the objective
\begin{equation}
  L \triangleq \sum_{\vv{a}\,\in\,\F_2^n}
  \Bigl(W_{f_{\vv{u}}}(\vv{a})^4 - 4\,W_{f_{\vv{u}}}(\vv{a})^3\Bigr).
  \label{eq:obj_L}
\end{equation}

\begin{proposition}[Lower Bound on $L$ for Even $n$]
\label{prop:lowerbound}
For even $n$, the minimum value of $L$ is
$L_{\min} = 2^{3n} - 4\cdot 2^{2n}$,
and $L$ is minimized when $f_{\vv{u}}:\F_2^n\to\{0,1\}$ with
$f_{\vv{u}}(\vv{0})=0$ is a bent function.
\end{proposition}

\begin{IEEEproof}
Write $W_{\vv{a}} \triangleq W_{f_{\vv{u}}}(\vv{a})$ for brevity,
and let $V \triangleq 2^{n/2}$, which is an even integer for $n$
even. Consider the witness polynomial
\begin{equation*}
  H(x) \triangleq (x^2 - V^2)\bigl((x-2)^2 - V^2\bigr),
\end{equation*}
whose roots are $\pm V$ and $2\pm V$. The polynomial $H$ is negative
only on the open intervals $(-V,\,-V+2)$ and $(V,\,V+2)$. Since
$W_{\vv{a}}$ is an even integer for every $\vv{a}$ (as shown in the
derivation of Proposition~\ref{prop:eigenvalue_formula}), and since
$(-V,\,-V+2)$ and $(V,\,V+2)$ each contain no even integers
(their only integer is odd), we conclude $H(W_{\vv{a}})\geq 0$ for
all $\vv{a}\in\F_2^n$.

Expanding $H(x)$:
\begin{equation*}
  H(x) = (x^4 - 4x^3) - (2V^2-4)x^2 + 4V^2 x + (V^4 - 4V^2) \geq 0,
\end{equation*}
which gives
\begin{equation*}
\begin{split}
  \sum_{\vv{a}} W_{\vv{a}}^4 - 4W_{\vv{a}}^3
  &\geq (2V^2-4)\sum_{\vv{a}}W_{\vv{a}}^2
    - 4V^2\sum_{\vv{a}}W_{\vv{a}}\\
  &\quad - 2^n(V^4 - 4V^2).
\end{split}
\end{equation*}
Applying the identities $\sum_{\vv{a}}W_{\vv{a}} = 2^n = V^2$ and
$\sum_{\vv{a}}W_{\vv{a}}^2 = 2^{2n} = V^4$ (Parseval's
identity~\cite{macwilliams1977}, using $f_{\vv{u}}(\vv{0})=0$):
\begin{align*}
  L &\geq (2V^2-4)V^4 - 4V^2 \cdot V^2 - V^2(V^4 - 4V^2) \\
    &= 2V^6 - 4V^4 - 4V^4 - V^6 + 4V^4 \\
    &= V^6 - 4V^4 = 2^{3n} - 4\cdot 2^{2n}.
\end{align*}
Equality holds iff $H(W_{\vv{a}})=0$ for all $\vv{a}$, i.e.,
$W_{\vv{a}}\in\{-V,\,-V+2,\,V,\,V+2\}$ for all $\vv{a}$.
For even $n$, bent functions $f_{\vv{u}}:\F_2^n\to\{0,1\}$ satisfy
$|W_{\vv{a}}|=V$ for all $\vv{a}$~\cite{macwilliams1977}, so
$W_{\vv{a}}\in\{-V,+V\}\subset\{-V,\,-V+2,\,V,\,V+2\}$, and
hence achieve the bound. Thus $L_{\min}=2^{3n}-4\cdot 2^{2n}$.
\end{IEEEproof}

\begin{proposition}[Minimum Eigenvalue and Optimal Eigenvector]
\label{prop:mineig}
For even $n$, the minimum eigenvalue of $A$ is
\begin{equation}
  \lambda_{\min}
  = -\frac{(2^n-1)(2^n-4)}{12},
  \label{eq:lambda_min}
\end{equation}
and an eigenvector in $\mathcal{Q}$ corresponding to $\lambda_{\min}$
is given by
$\vv{e}_{\vv{u}} = \bigl[(-1)^{\vv{u}\cdot\vv{v}_1},\,
(-1)^{\vv{u}\cdot\vv{v}_2},\,\ldots,\,
(-1)^{\vv{u}\cdot\vv{v}_{|\mathcal{C}_e|}}\bigr]^\top$,
where $\vv{u}$ is the punctured truth table of a bent function
$f_{\vv{u}}:\F_2^n\to\{0,1\}$ with $f_{\vv{u}}(\vv{0})=0$.
\end{proposition}

\begin{IEEEproof}
By Proposition~\ref{prop:lowerbound}, $L$ is minimized when
$\vv{u}$ is the punctured truth table of a bent function
$f_{\vv{u}}$ with $f_{\vv{u}}(\vv{0})=0$, achieving
$L_{\min}=2^{3n}-4\cdot 2^{2n}$. Substituting into
\eqref{eq:eigenvalue_formula}:
\begin{align*}
  24\,\lambda_{\min}
  &= \frac{1}{2^n}(2^{3n} - 4\cdot 2^{2n})
     - 3\cdot 2^{2n} + 14\cdot 2^n - 8 \\
  &= 2^{2n} - 4\cdot 2^n - 3\cdot 2^{2n} + 14\cdot 2^n - 8 \\
  &= -2\cdot 2^{2n} + 10\cdot 2^n - 8 \\
  &= -2(2^{2n} - 5\cdot 2^n + 4)
   = -2(2^n-1)(2^n-4),
\end{align*}
which gives $\lambda_{\min} = -(2^n-1)(2^n-4)/12$.
One can verify that $(2^n-1)(2^n-4)$ is always divisible by $12$
for even $n$, so $\lambda_{\min}$ is an integer.
It follows from Corollary~\ref{cor:eigvecs} that the character
vector $\vv{e}_{\vv{u}}\in\mathcal{Q}$ corresponding to such
$\vv{u}$ is an eigenvector of $A$ with eigenvalue $\lambda_{\min}$.
\end{IEEEproof}

\begin{theorem}[Optimal Parity Assignment for Even $n$]
\label{thm:opt_parity_even}
For even $n$, the cut vector
$\vv{z} = \bigl[(-1)^{\vv{u}\cdot\vv{v}_1},\,
(-1)^{\vv{u}\cdot\vv{v}_2},\,\ldots,\,
(-1)^{\vv{u}\cdot\vv{v}_{|\mathcal{C}_e|}}\bigr]^\top$,
where $\vv{u}$ is the punctured truth table of a bent function
$f_{\vv{u}}:\F_2^n\to\{0,1\}$ with $f_{\vv{u}}(\vv{0})=0$, defines
a parity assignment for $\mathcal{C}_e$ that satisfies
constraint~\eqref{eq:opt_pairs}.
\end{theorem}

\begin{IEEEproof}
Follows directly from Theorems~\ref{thm:maxcut}
and~\ref{thm:eigvec_opt} and Proposition~\ref{prop:mineig}.
\end{IEEEproof}

\begin{remark}
\label{rem:other_minimizers}
Bent functions with $f_{\vv{u}}(\vv{0})=0$ constitute only one class
of minimizers for $L$. As seen from the proof of
Proposition~\ref{prop:lowerbound}, any $f_{\vv{u}}$ with
$f_{\vv{u}}(\vv{0})=0$ whose Walsh spectrum lies entirely in
$\{-V,\,-V+2,\,V,\,V+2\}$ achieves $L_{\min}$, of which the
condition $|W_{\vv{a}}|=V$ for all $\vv{a}$ (the bent function
condition) is a strict special case. Consequently, there may exist
additional eigenvectors in $\mathcal{Q}$ corresponding to
$\lambda_{\min}$ beyond those parameterized by bent functions,
and thus potentially further parity assignments satisfying
constraint~\eqref{eq:opt_pairs} that are not fully characterized
here.
\end{remark}


\subsection{Optimal Parity Assignment for $\mathcal{C}_o$ and
Complete SEFCC Construction}
\label{subsec:Co_construction}

Having established optimal parity assignments for $\mathcal{C}_e$,
we now turn to $\mathcal{C}_o$. We show that the distance-$4$ graph
$G_4' = (\mathcal{C}_o, \mathcal{E}_4')$ on $\mathcal{C}_o$ is
isomorphic to $G_4$, which allows us to transfer the spectral
characterization of Section~\ref{subsec:eigenvalues} directly to
$\mathcal{C}_o$.

We first note that Theorem~\ref{thm:maxcut} extends verbatim to
$\mathcal{C}_o$:

\begin{theorem}[Max-Cut Formulation for $\mathcal{C}_o$]
\label{thm:maxcut_Co}
Let $\vv{w}_1,\ldots,\vv{w}_{|\mathcal{C}_o|}$ be the HCWs in
$\mathcal{C}_o$, where $|\mathcal{C}_o|=2^{2^n-n-2}$. Let
$G_4'=(\mathcal{C}_o,\mathcal{E}_4')$ be the graph with
$\{\vv{w}_i,\vv{w}_j\}\in\mathcal{E}_4'$ iff
$\dH(\vv{w}_i,\vv{w}_j)=4$, and let $A'$ be the adjacency matrix
of $G_4'$. Given that parity assignments to HCWs in $\mathcal{C}_o$ are drawn from $\{\vv{p'},\bar{\vv{p'}}\}\in\{P_e,P_o\}$, define the cut vector $\vv{z}\in\{+1,-1\}^{|\mathcal{C}_o|}$ by $z_i=+1$ if $\pC(\vv{w}_i)=\vv{p'}$ and $z_i=-1$ if $\pC(\vv{w}_i)=\bar{\vv{p'}}$. Then the number of same-parity pairs in $\mathcal{E}_4'$ is minimized if and only if $\vv{z}$ minimizes $\vv{z}^\top A'\vv{z}$.
\end{theorem}

\begin{IEEEproof}
Identical to the proof of Theorem~\ref{thm:maxcut}.
\end{IEEEproof}

\begin{lemma}[Isomorphism of $G_4$ and $G_4'$]
\label{lem:isomorphism}
The graphs $G_4$ and $G_4'$ are isomorphic. Specifically, for any
fixed $\vv{w}'\in\mathcal{C}_o$, the map
$\phi:\mathcal{C}_e\to\mathcal{C}_o$ defined by
$\phi(\vv{v})=\vv{v}\oplus\vv{w}'$ is a graph isomorphism.
\end{lemma}

\begin{IEEEproof}
We verify the three properties of a graph isomorphism.

\textit{Well-definedness:} For any $\vv{v}\in\mathcal{C}_e$,
$\wt(\vv{v}\oplus\vv{w}')\equiv\wt(\vv{v})+\wt(\vv{w}')\equiv 0+1
\equiv 1\pmod{2}$, so $\phi(\vv{v})\in\mathcal{C}_o$.

\textit{Bijectivity:} The map is injective since $\oplus$ is
cancellative. It is surjective since for any $\vv{x}\in\mathcal{C}_o$,
the vector $\vv{v}=\vv{x}\oplus\vv{w}'$ satisfies
$\wt(\vv{v})\equiv\wt(\vv{x})+\wt(\vv{w}')\equiv 0\pmod{2}$, so
$\vv{v}\in\mathcal{C}_e$, and $\phi(\vv{v})=\vv{x}$.

\textit{Adjacency preservation:} For any
$\vv{v}_1,\vv{v}_2\in\mathcal{C}_e$,
\begin{equation*}
\begin{split}
    \dH(\phi(\vv{v}_1),\phi(\vv{v}_2))
  = \wt((\vv{v}_1\oplus\vv{w}')\oplus(\vv{v}_2\oplus\vv{w}'))
  \\= \wt(\vv{v}_1\oplus\vv{v}_2)
  = \dH(\vv{v}_1,\vv{v}_2),
\end{split}
\end{equation*}
so $\phi$ preserves distances and hence adjacency.
\end{IEEEproof}

\begin{corollary}[Spectral Consequences of Isomorphism]
\label{cor:isomorphism}
Fix an ordering $\vv{w}_1,\ldots,\vv{w}_{|\mathcal{C}_o|}$ of the
HCWs in $\mathcal{C}_o$ and $\vv{v}_1,\ldots,\vv{v}_{|\mathcal{C}_e|}$
of those in $\mathcal{C}_e$, and let $P$ be the permutation matrix
defined by $P_{ik}=1$ if $\vv{w}_i=\phi(\vv{v}_k)$ and $0$
otherwise. Then:
\begin{enumerate}[nosep, topsep=0pt, partopsep=0pt]
\item[a)] $A' = PAP^\top$;
\item[b)] $A$ and $A'$ have the same characteristic polynomial and
hence the same eigenvalues, including $\lambda_{\min}$;
\item[c)] If $\vv{e}$ is an eigenvector of $A$ corresponding to
eigenvalue $\lambda$, then $P\vv{e}$ is an eigenvector of $A'$
corresponding to the same eigenvalue $\lambda$, and vice versa.
\end{enumerate}
\end{corollary}

\begin{IEEEproof}
Part~(a) follows directly from the definition of $P$ and the
isomorphism $\phi$. Parts~(b) and~(c) are standard consequences of
matrix similarity~\cite{macwilliams1977}: since $A'=PAP^\top$ and
$P$ is orthogonal ($P^\top=P^{-1}$), $A$ and $A'$ are similar and
hence share their characteristic polynomial and eigenvalues. If
$A\vv{e}=\lambda\vv{e}$, then
$A'(P\vv{e})=PAP^\top P\vv{e}=PA\vv{e}=\lambda(P\vv{e})$.
\end{IEEEproof}

Using Corollary~\ref{cor:isomorphism}, we extend
Theorem~\ref{thm:eigvec_opt} to $\mathcal{C}_o$.

\begin{theorem}[Optimal Parity Assignment for $\mathcal{C}_o$ via
Eigenvectors]
\label{thm:eigvec_opt_Co}
Let $\mathcal{Q}'\triangleq\{P\vv{e}_{\vv{u}}\mid
\vv{e}_{\vv{u}}\in\mathcal{Q}\}=P\mathcal{Q}$ be the set of orthogonal
eigenvectors of $A'$ obtained by applying $P$ to each element of
$\mathcal{Q}$. Any vector $\vv{e}_{\vv{u}}'\in\mathcal{Q}'$ or its
negative $-\vv{e}_{\vv{u}}'$ corresponding to $\lambda_{\min}$ is a
valid cut vector $\vv{z}\in\{+1,-1\}^{|\mathcal{C}_o|}$ and minimizes
$\vv{z}^\top A'\vv{z}$. Such vectors yield the optimal parity
assignment for $\mathcal{C}_o$, i.e., a parity assignment satisfying
constraint~\eqref{eq:opt_pairs}.
\end{theorem}

\begin{IEEEproof}
By Corollary~\ref{cor:isomorphism}(c), each
$\vv{e}_{\vv{u}}'\triangleq P\vv{e}_{\vv{u}}\in\mathcal{Q}'$ is an
eigenvector of $A'$ corresponding to the same eigenvalue
$\lambda_{\vv{u}}$ as $\vv{e}_{\vv{u}}$ in $A$. Since $P$ is a
permutation matrix, it only reorders the entries of $\vv{e}_{\vv{u}}$;
as $\vv{e}_{\vv{u}}\in\{+1,-1\}^{|\mathcal{C}_e|}$, we have
$\vv{e}_{\vv{u}}'=P\vv{e}_{\vv{u}}\in\{+1,-1\}^{|\mathcal{C}_o|}$,
so it is a valid cut vector. Orthogonality of $\mathcal{Q}'$ follows
from that of $\mathcal{Q}$: since $P$ is orthogonal,
$\langle\vv{e}_{\vv{u}}',\vv{e}_{\vv{u}'}'\rangle =
\langle P\vv{e}_{\vv{u}}, P\vv{e}_{\vv{u}'}\rangle =
\langle\vv{e}_{\vv{u}},\vv{e}_{\vv{u}'}\rangle = 0$
for $\vv{u}\neq\vv{u}'$. The rest follows from
Theorem~\ref{thm:eigvec_opt} and Corollary~\ref{cor:isomorphism}(b),
which guarantee that $\lambda_{\min}$ is the same for both $A$
and $A'$.
\end{IEEEproof}

\begin{theorem}[Optimal Parity Assignment for $\mathcal{C}_o$,
Even $n$]
\label{thm:opt_Co_even}
Fix orderings $\vv{v}_1,\ldots,\vv{v}_{|\mathcal{C}_e|}$ and
$\vv{w}_1,\ldots,\vv{w}_{|\mathcal{C}_o|}$, and fix the isomorphism
$\phi(\vv{v})=\vv{v}\oplus\vv{w}'$ for some $\vv{w}'\in\mathcal{C}_o$.
Let $P$ be the permutation matrix with respect to these orderings.
For even $n$, let $\vv{u}$ be the punctured truth table of any
bent function $f_{\vv{u}}:\F_2^n\to\{0,1\}$ with
$f_{\vv{u}}(\vv{0})=0$. Then the cut vector
\begin{equation}
  \vv{z} = P\,\bigl[(-1)^{\vv{u}\cdot\vv{v}_1},\,
  (-1)^{\vv{u}\cdot\vv{v}_2},\,\ldots,\,
  (-1)^{\vv{u}\cdot\vv{v}_{|\mathcal{C}_e|}}\bigr]^\top
  \label{eq:z_Co}
\end{equation}
defines a parity assignment for $\mathcal{C}_o$ satisfying
constraint~\eqref{eq:opt_pairs}.
\end{theorem}

\begin{IEEEproof}
Follows directly from Theorems~\ref{thm:maxcut_Co},
\ref{thm:eigvec_opt_Co}, and Proposition~\ref{prop:mineig}.
\end{IEEEproof}

\begin{remark}[Complete SEFCC Construction]
\label{rem:construction}
Theorems~\ref{thm:bipartite_parity},
\ref{thm:opt_parity_even}, and~\ref{thm:opt_Co_even} together
describe a complete procedure for constructing an optimal SEFCC
$\mathcal{C}:\F_2^{2^n-1}\to\F_2^{2^n+1}$ satisfying
constraints~\eqref{eq:sefcc_condition},~\eqref{eq:opt_dmin},
and~\eqref{eq:opt_pairs} for even $n$:
\begin{enumerate}[nosep, topsep=2pt, partopsep=0pt]
\item Assign $\mathcal{C}_e$ parities from $P_e$ (without loss of
generality) and $\mathcal{C}_o$ parities from $P_o$.
\item For even $n$, compute $\vv{u}$ as the punctured truth
table of any bent function $f_{\vv{u}}:\F_2^n\to\{0,1\}$ with
$f_{\vv{u}}(\vv{0})=0$. The cut vector
$\vv{e}_{\vv{u}}=[(-1)^{\vv{u}\cdot\vv{v}_i}]_{i=1}^{|\mathcal{C}_e|}$
assigns each $\vv{v}_i\in\mathcal{C}_e$ parity $\vv{p}\in P_e$ if
$(e_{\vv{u}})_i=+1$ and parity $\bar{\vv{p}}$ if
$(e_{\vv{u}})_i=-1$.
\item Apply $P$ to obtain the cut vector $P\vv{e}_{\vv{u}}$ for
$\mathcal{C}_o$, assigning parities from $P_o$ analogously.
\item Assign each NHCW $\vv{u}_{i,j}$ the parity
$\overline{\pC(\vv{c}_i)}$ per Theorem~\ref{thm:valid_dmin2}(a).
\end{enumerate}
For odd $n$, the problem is reduced to finding $f_{\vv{u}}'$ minimizing~\eqref{eq:obj_L}, and then the same procedure works with $f_{\vv{u}}'$ in step 2).
\end{remark}


\section{Conclusion}
\label{sec:conc}
This paper studied SEFCCs with optimal data protection for the
$[2^n-1,\,2^n-1-n,\,3]$-Hamming code membership function for
general $n \geq 2$. The validity conditions
of~\cite{durgi2026hcmf} were extended to the general HCMF, and
the bipartite structure of the distance-$3$ codeword graph was
established for all $n \geq 2$. Since the construction
of~\cite{durgi2026hcmf} does not generalise, a novel framework
was developed that reduces optimal SEFCC design to a max-cut
problem on the distance-$4$ graphs of the two partite sets. It was shown that certain eigenvectors corresponding to the minimum eigenvalues of these graphs are representational of SEFCCs satisfying constraints~\eqref{eq:opt_dmin} and~\eqref{eq:opt_pairs}.
The minimum eigenvalue of one of these graphs and its eigenvectors were
characterised via an optimisation problem over Walsh coefficients,
and a tight lower bound was derived and shown to be attained by
bent functions for even $n \geq 2$, establishing a precise
connection between optimal SEFCC design and the theory of bent
Boolean functions. The eigenvectors of the second graph were derived easily with the help of those of the first graph. The solution specialised to $n=3$ by solving the optimization problem manually subsumes
the construction of~\cite{durgi2026hcmf}. For general odd $n$, the
characterization of $\lambda_{\min}$ and the corresponding
optimal eigenvectors (the analogue of
Proposition~\ref{prop:lowerbound} for odd $n$) remains an open
optimization problem and the corresponding optimal eigenvectors and parity assignments remain uncharacterized.


\appendices
\section{Proof that $\mathcal{C}_e^\perp = \mathcal{RM}^*(1,n)$}
\label{appendix:A}

We show that the dual of $\mathcal{C}_e$ is the punctured first-order
Reed--Muller code $\mathcal{RM}^*(1,n)$.

\textit{Reed--Muller code $\mathcal{RM}(1,n)$:} The first-order
Reed--Muller code $\mathcal{RM}(1,n)$ is the binary linear code of
length $2^n$ whose codewords are the truth tables of all affine
boolean functions on $\F_2^n$, i.e.,
\begin{equation*}
  \mathcal{RM}(1,n)
  = \bigl\{\vv{v}_f \mid f(\vv{x}) = \vv{a}\cdot\vv{x} \oplus c,\;
    \vv{a}\in\F_2^n,\; c\in\{0,1\}\bigr\},
\end{equation*}
where $\vv{v}_f$ denotes the truth table of $f$. The dimension of
$\mathcal{RM}(1,n)$ is $n+1$, since the set of functions
$\{1, x_1, x_2, \ldots, x_n\}$ forms a basis for the vector space
of affine functions on $\F_2^n$.

\textit{Punctured code $\mathcal{RM}^*(1,n)$:} The punctured
first-order Reed--Muller code $\mathcal{RM}^*(1,n)$ is obtained from
$\mathcal{RM}(1,n)$ by deleting the coordinate corresponding to the
evaluation at $\vv{x}=\vv{0}$. It has length $2^n-1$ and dimension
$n+1$, and its codewords are the punctured truth tables of all affine
functions on $\F_2^n$.

\textit{Dual of the Hamming code $\mathcal{C}_H$:} Let $H$ be the
$n\times(2^n-1)$ parity check matrix of $\mathcal{C}_H$, so that
$\mathcal{C}_H=\ker(H)$. Then
\begin{equation*}
  \mathcal{C}_H^\perp
  = \ker(H)^\perp
  = \mathrm{Col}(H^\top)
  = \mathrm{Row}(H).
\end{equation*}
For $\vv{a}\in\F_2^n$, the vector $\vv{y}=H^\top\vv{a}$ has $i$-th
coordinate $y_i=\vv{a}\cdot\vv{\alpha}_i$, where $\vv{\alpha}_i$
denotes the $i$-th column of $H$. Since the columns
$\vv{\alpha}_1,\ldots,\vv{\alpha}_{2^n-1}$ of a Hamming code's parity
check matrix are exactly all the nonzero vectors in $\F_2^n$, the
vector $\vv{y}$ is precisely the punctured truth table of the linear
function $\ell_{\vv{a}}(\vv{x})=\vv{a}\cdot\vv{x}$. Hence
$\mathcal{C}_H^\perp$ consists of the punctured truth tables of all
linear functions on $\F_2^n$, i.e., $\mathcal{C}_H^\perp$ is the
punctured simplex code.

\textit{Structure of $\mathcal{C}_e$:} Let
$E\triangleq\{\vv{x}\in\F_2^{2^n-1}\mid\wt(\vv{x})\equiv 0\pmod{2}\}$
be the even-weight subspace of $\F_2^{2^n-1}$. Note that
$\vv{x}\in E$ iff $\vv{x}\cdot\bar{\vv{1}}=0$, where
$\bar{\vv{1}}\in\F_2^{2^n-1}$ is the all-ones vector. Hence
$E=\langle\bar{\vv{1}}\rangle^\perp$, and
$\mathcal{C}_e=\mathcal{C}_H\cap E$.

\textit{Computing $\mathcal{C}_e^\perp$:} By the fundamental identity
$(A\cap B)^\perp = A^\perp + B^\perp$ for linear subspaces $A,B$ of
a finite-dimensional inner product space~\cite{macwilliams1977},
\begin{equation*}
  \mathcal{C}_e^\perp
  = (\mathcal{C}_H\cap E)^\perp
  = \mathcal{C}_H^\perp + E^\perp
  = \mathcal{C}_H^\perp + \langle\bar{\vv{1}}\rangle.
\end{equation*}
Any $\vv{v}\in\mathcal{C}_e^\perp$ is of the form
$\vv{u}\oplus c\bar{\vv{1}}$ for some $\vv{u}\in\mathcal{C}_H^\perp$
and $c\in\{0,1\}$. Since $\vv{u}$ is the punctured truth table of
$\vv{a}\cdot\vv{x}$ for some $\vv{a}\in\F_2^n$, the $i$-th
coordinate of $\vv{v}$ is
\begin{equation*}
  v_i = \vv{a}\cdot\vv{\alpha}_i \oplus c,
\end{equation*}
which is the evaluation at $\vv{\alpha}_i$ of the affine function
$f(\vv{x})=\vv{a}\cdot\vv{x}\oplus c$. Therefore $\vv{v}$ is the
punctured truth table of the affine function $\vv{a}\cdot\vv{x}\oplus c$,
and $\mathcal{C}_e^\perp$ contains the punctured truth tables of all
$2^{n+1}$ affine boolean functions on $\F_2^n$. Since
$|\mathcal{C}_e^\perp|=2^{n+1}=|\mathcal{RM}^*(1,n)|$, we conclude
\begin{equation*}
  \mathcal{C}_e^\perp = \mathcal{RM}^*(1,n). 
\end{equation*}

\section{Proof that $24\,K_4(w;\,N) = X^4 + (8-6N)X^2 + 3N^2 - 6N$}
\label{appendix:B}

We derive a closed-form polynomial expression for the degree-$4$
Krawtchouk polynomial $K_4(w;\,N)$. Recall that the Krawtchouk
polynomial of degree $h$ is defined as
\begin{equation*}
  K_h(w;\,N)
  = \sum_{\substack{\vv{x}\in\F_2^N \\ \wt(\vv{x})=h}}(-1)^{\vv{y}\cdot\vv{x}},
  \quad \wt(\vv{y})=w,
\end{equation*}
which admits the explicit form
\begin{equation*}
  K_h(w;\,N)
  = \sum_{j=0}^{h}(-1)^j\binom{w}{j}\binom{N-w}{h-j}.
\end{equation*}
Note the symmetry property $K_h(w;\,N)=(-1)^h K_h(N-w;\,N)$.
The generating function of $K_h(w;\,N)$ in $h$ is
\begin{equation}
  G(z) = (1-z)^w(1+z)^{N-w}
  = \sum_{h=0}^{N} K_h(w;\,N)\,z^h,
  \label{eq:genfun}
\end{equation}
i.e., $K_h(w;\,N)$ is the coefficient of $z^h$ in $G(z)$.

\textit{Substitution:} Set $X \triangleq N-2w$, so that
$w = \tfrac{N-X}{2}$ and $N-w = \tfrac{N+X}{2}$. Then
\begin{align*}
  G(z)
  &= (1-z)^{\frac{N-X}{2}}(1+z)^{\frac{N+X}{2}} \\
  &= (1-z^2)^{N/2}\left(\frac{1+z}{1-z}\right)^{X/2} \\
  &= (1-z^2)^{N/2}
     \exp\!\left(\frac{X}{2}\ln\frac{1+z}{1-z}\right).
\end{align*}
Using the Taylor expansion
$\tfrac{1}{2}\ln\tfrac{1+z}{1-z} = z + \tfrac{z^3}{3} + \cdots$,
\begin{equation}
  G(z) = (1-z^2)^{N/2}
  \exp\!\left(Xz + \frac{Xz^3}{3} + \cdots\right).
  \label{eq:Gexpanded}
\end{equation}

\textit{Expanding each factor up to $z^4$:}
\begin{align*}
  (1-z^2)^{N/2}
  &= 1 - \frac{N}{2}z^2
     + \frac{(N/2)(N/2-1)}{2}z^4 + O(z^6) \\
  &= 1 - \frac{N}{2}z^2
     + \frac{N^2-2N}{8}z^4 + O(z^6),
\end{align*}
\begin{align*}
  \exp\!\left(Xz + \frac{Xz^3}{3}\right)
  &= 1
   + \left(Xz + \frac{Xz^3}{3}\right)
   + \frac{1}{2}\left(Xz + \frac{Xz^3}{3}\right)^2 \\
  &\quad + \frac{(Xz)^3}{6}
   + \frac{(Xz)^4}{24}
   + O(z^5).
\end{align*}
Collecting the relevant terms up to $z^4$ from the exponential:
\begin{align*}
  [z^0] &: 1, \\
  [z^2] &: \frac{X^2}{2}, \\
  [z^4] &: \frac{X^4}{24} + \frac{X^2}{3}.
\end{align*}

\textit{Extracting the coefficient of $z^4$:}
Multiplying the two factors and collecting the coefficient of $z^4$:
\begin{align*}
  [z^4]\,G(z)
  &= 1\cdot\left(\frac{X^4}{24}+\frac{X^2}{3}\right)
   + \left(-\frac{N}{2}\right)\cdot\frac{X^2}{2}
   + \frac{N^2-2N}{8}\cdot 1 \\
  &= \frac{X^4}{24}
   + \frac{X^2}{3}
   - \frac{NX^2}{4}
   + \frac{N^2-2N}{8}.
\end{align*}
Multiplying through by $24$:
\begin{align*}
  24\,K_4(w;\,N)
  &= X^4 + 8X^2 - 6NX^2 + 3N^2 - 6N \\
  &= X^4 + (8-6N)X^2 + 3N^2 - 6N,
\end{align*}
where $X = N - 2w$. 

\section*{Acknowledgment}
This work was supported partly by the seed grant provided by the Indian Institute of Technology, Hyderabad to Anjana A. Mahesh.

\bibliographystyle{IEEEtran}

\end{document}